\documentclass[11pt]{article}
\usepackage[margin=1in]{geometry}
\usepackage{amsmath,amssymb,amsthm}
\usepackage{mathrsfs}
\usepackage{xspace}
\usepackage{algorithm}
\usepackage{algpseudocode}
\usepackage{tikz}
\usepackage{multirow}
\theoremstyle{definition}
\newtheorem{definition}{Definition}
\newtheorem{theorem}{Theorem}
\newtheorem{lemma}[theorem]{Lemma}
\newtheorem{claim}[theorem]{Claim}

\newcommand{\cov}{\textsf{cov}\xspace}
\newcommand{\val}{\textsf{val}\xspace}
\newcommand{\flow}{\textsf{flow}\xspace}
\newcommand{\pkso}{\textsf{P$k$SO}\xspace}
\newcommand{\wkpp}{\textsf{W$k$PP}\xspace}
\newcommand{\pknapso}{\textsf{PKnapSO}\xspace}
\newcommand{\wknappp}{\textsf{WKnapPP}\xspace}
\newcommand{\pcov}{\mathscr{P}_{\cov}}
\newcommand{\pcks}{\textsf{PC$k$S}\xspace}
\newcommand{\wckpp}{\textsf{WC$k$PP}\xspace}
\newcommand{\upcks}{\textsf{UPC$k$S}\xspace}

\begin{document}

\title{Colorful Priority $k$-Supplier}
\author{
Chandra Chekuri\thanks{Dept. of Computer Science, University of Illinois, Urbana-Champaign, IL 61820. Work on this paper partially supported by NSF grant CCF-1910149. {\tt chekuri@illinois.edu}}
\and
Junkai Song\thanks{The University of Hong Kong. Part of this work was done while visiting University of Illinois, Urbana-Champaign. {\tt dsgsjk@connect.hku.hk}}
}
\date{\today}

\maketitle

\begin{abstract}
In the \emph{Priority $k$-Supplier} problem the input consists of a metric space $(F \cup C, d)$ over set of facilities $F$ and a set of clients $C$, an integer $k > 0$, and 
a non-negative radius $r_v$ for each client $v \in C$. The goal is to select $k$ facilities
$S \subseteq F$ to minimize $\max_{v \in C} \frac{d(v,S)}{r_v}$ where $d(v,S)$ is the distance of $v$ to the closes facility in $S$. This problem generalizes the
well-studied $k$-Center and $k$-Supplier problems, and admits a $3$-approximation \cite{plesnik1987heuristic,bajpai2022revisiting}. In this paper we consider two outlier versions. 
The \emph{Priority $k$-Supplier with Outliers} problem \cite{bajpai2022revisiting} allows a specified number of outliers to be uncovered, and the \emph{Priority Colorful $k$-Supplier} problem is a further generalization where clients are partitioned into $c$ colors and each color class allows a specified number of outliers. These problems are partly motivated by recent interest in fairness in clustering and other optimization problems involving algorithmic decision making.

We build upon the work of \cite{bajpai2022revisiting} and improve their $9$-approximation \emph{Priority $k$-Supplier with Outliers} problem to a $1+3\sqrt{3}\approx 6.196$-approximation. For the \emph{Priority Colorful $k$-Supplier} problem, we present the first set of approximation algorithms. For the general case with $c$ colors, we achieve a $17$-pseudo-approximation using $k+2c-1$ centers. For the setting of $c=2$, we obtain a 
$7$-approximation in random polynomial time, and
a $2+\sqrt{5}\approx 4.236$-pseudo-approximation using $k+1$ centers.
\end{abstract}

\section{Introduction}
The discrete $k$-Center problem is a well-studied clustering problem with several applications. The input consists a finite metric space $(X,d)$ and an integer $k$. The goal is
to choose $k$ centers $S \subseteq X$ to minimize $\max_{v \in X} d(v,S)$ where
$d(v,S)$ is the distance of $v$ to the set $S$ (defined as $\min_{u \in S} d(u,v)$).
This is an NP-Hard problem that admits a $2$-approximation algorithm \cite{hochbaum1986unified,gonzalez85} and moreover this approximation ratio is tight assuming $P \neq NP$ \cite{hochbaum1986unified}. The $k$-Supplier problem is
a generalization in which $X = F \uplus C$ with $F$ denoting a set of facilities
and $C$ denoting a set of clients. Now the goal is to find a set $S$ of $k$ facilities
to minimize $\max_{v \in C} d(v,S)$. The $k$-Supplier problem admits a $3$-approximation and this factor is also tight under $P \neq NP$ \cite{hochbaum1986unified}.
Plesn\'ik \cite{plesnik1987heuristic} considered \emph{priority} versions of these
problems. The input now consists of radius requirement $r_v$ for each client $v \in C$. The goal is to select centers to minimize $\max_{v \in C} \frac{d(v,S)}{r_v}$.~\footnote{Plesn\'ik used the terminology ``weighted" $k$-Center for his problem.  \cite{bajpai2022revisiting} et al.\ used the terminology of ``priority" to avoid confusion with another problem, and due to the reformulation via the notion of client
radii which correspond to the inverse of weights.} 
Plean\'ik generalized the $2$-approximation algorithm for $k$-Center of \cite{hochbaum1986unified} to Priority $k$-Center. Bajpai et al. \cite{bajpai2022revisiting} recently observed that Plesn\'ik's algorithm and analysis
can be extended to obtain a $3$-approximation for Priority $k$-Supplier
and some further generalizations that constrain the centers in more complex ways (matroid 
and knapsack constraints).

Bajpai et al. \cite{bajpai2022revisiting} also considered the \emph{outlier} versions of priority clustering problems. They were motivated by recent interest in 
incorporating fairness in clustering objectives.
In the \emph{Priority $k$-Supplier with Outliers} (\pkso) problem \cite{bajpai2022revisiting}, an additional parameter $m$ is given and the goal is to cover at least $m$ clients. That is, minimize $\alpha$ such that for at least $m$ clients $v\in C$, $d(v,S)\leq \alpha\cdot r_v$. In this paper
we consider a further generalization, the \emph{Priority Colorful $k$-Supplier} (\pcks) problem. In this version, the client set $C$ is partitioned into $c$ colors $\{C_1,C_2,\dots,C_c\}$ and $c$ integer parameters $m_1,m_2,\dots,m_c$ are given. The goal is to cover at least $m_i$ clients from each color $i$. That is, minimize $\alpha$ such that for at least $m_i$ clients $v\in C_i$, $d(v,S)\leq \alpha\cdot r_v$ for any $1\leq i\leq c$. The colorful version of $k$-Center was introduced in 
\cite{bandyapadhyay2019constant} and has seen several results \cite{anegg2022techniques,jia2022fair}.
As far as we know, \pcks has not been formally studied in previous work.

For the $k$-Center with Outlier problem, Charikar et al~\cite{charikar2001algorithms} obtained a clever greedy $3$-approximation algorithm. The approximation ratio was later improved to $2$ due to Chakrabarty et al~\cite{chakrabarty2020non} using an LP-based approach. The Colorful $k$-center probelm is proposed by Bandyapadhyay et al~\cite{bandyapadhyay2019constant} and they introduced a pseudo-approximation algorithm that yields a $3$-approximation using at most $k+c-1$ centers. Subsequently, Jia et al~\cite{jia2022fair} obtained a true $3$-approximation algorithm when 
$c$ is a fixed constant. Anegg et al~\cite{anegg2022techniques} obtain an $O(1)$-approximation algorithm when $c$ is a fixed constant for generalizations including
matroid and knapsack versions.

The Priority $k$-Supplier with Outliers (\pkso) problem is studied in~\cite{bajpai2022revisiting}. In the paper, the authors made use of the LP-based framework from~\cite{chakrabarty2020non,10.1145/3338513} and obtained a $9$-approximation algorithm. Futhermore, they showed that this result can be extened to the matroid and the knapsack settings, obtaining a $9$-approximation and a $17$-approximation respectively\footnote{The approximation ratio for the knapsack version claimed in \cite{bajpai2022revisiting} is $14$ but there is a minor error in the analysis
that when corrected yields a slightly weaker bound of $17$. See Section~\ref{sec:pknapso} for detailed discussion.}.  In terms of hardness of approximation, we only know a hardness factor of $3$ which is inherited from the $k$-Supplier problem. There is a natural LP relaxation for the problem and no factor worse than $3$ is known on its integrality gap.

\paragraph{Results and ideas:} For \pkso we improve upon the algorithm in~\cite{bajpai2022revisiting}, deriving $1+3\sqrt{3}\approx6.196$-approximation. \cite{bajpai2022revisiting} considered
special cases when the number of distinct client radii is small. They obtained a $3$-approximation for two distinct radii (which is tight), 
and a $5$-approximation for three radii. We improve the approximation ratio for three radii
from $5$ to $3.94$. When the radii are powers of $b$, the approximation ratio is improved from $\frac{3b-1}{b-1}$ to $\frac{3b^2-1}{b^2-1}$. 

For \pcks there were no previous results. We derive a $17$-pseudo-approximation 
that uses $k+2c-1$ centers. For a special case of \pcks where there are only $2$ colors and the vertices with the same color has a same radius, we obtain a $2+\sqrt{5}\approx4.236$-pseudo-approximation using $k+1$ centers. In addition, based on the framework due to~\cite{anegg2022techniques}, we obtain a true $7$-approximation in randomized polynomial time.

Our results build on the framework in \cite{bajpai2022revisiting} that has several ingredients: (i) solving an LP relaxation to obtain a fractional solution (ii) rounding the client radii to powers of $b$ for some parameter $b > 1$ to create well-separated radius classes (iii) pre-processing the clients in each radius class based on the LP solution to create well-separated cluster centers in each radius class (iv) creating a contact DAG on the cluster centers, distances between them, and the coverage information from the LP solution, and solving an auxiliary flow problem that allows the algorithm to identify the centers that form the final output. For the knapsack constraint, instead of using a contact DAG, the framework creates a contact forest and uses dynamic programming on the forest and a cut-and-round approach to solve the underlying LP relaxation --- this more complex approach is needed since the underlying flow problem on the contact DAG would not yield an integer polytope/solution. For \pkso we obtain an improved approximation by altering the way in which the contact DAG is created and balancing the parameters based on this different approach. We note that this optimization is specific to \pkso and does not extend to the the more general matroid constraint. For \pcks we build on the the insight from \cite{bandyapadhyay2019constant} who showed that one can use the LP solution to process the given instance to obtain a feasible solution to a related LP that has a small number of constraints; one can then take a basic feasible solution of this new LP which yields a pseudo-approximation for the original problem. In our setting the priorities create a non-trivial challenge. We use the LP solution to a create a contact forest (we borrow the approach from the knapsack version of \pkso that we already mentioned) --- although the approximation ratio is worse if we use a contact forest instead of a contact DAG, we take advantage of the forest structure to show the existence of a basic feasible solution 
to an auxiliary LP that has only a few fractional values. This gives us the desired pseudo-approximation.

\smallskip
Tabel~\ref{tbl:results} summarizes the approximation ratios from the past work and this paper.

\begin{table}[ht]
    \label{tbl:results}
    \begin{center}
    \begin{tabular}{|c|cc|}   
    \hline   \textbf{Problems} & \multicolumn{1}{c|}{\textbf{Prior}} & \textbf{This Paper} \\
\hline   Priority $k$-Center & \multicolumn{1}{c|}{$2$~\cite{plesnik1987heuristic}} & \\
    \hline   Priority $k$-Supplier & \multicolumn{1}{c|}{$3$~\cite{bajpai2022revisiting}} & \\
    \hline   $k$-Center with Outliers & \multicolumn{1}{c|}{$2$~\cite{chakrabarty2020non}} & \\
    \hline   $k$-Supplier with Outliers & \multicolumn{1}{c|}{$3$~\cite{chakrabarty2020non}} & $3$ \\
    \hline   Priority $k$-Supplier with Outliers & & \\
    \hline   General & \multicolumn{1}{c|}{$9$~\cite{bajpai2022revisiting}} & $1+3\sqrt{3}\approx 6.196$ \\
    \hline 2 types of radii  & \multicolumn{1}{c|}{$3$~\cite{bajpai2022revisiting}} & $3$ \\
    \hline 3 types of radii  & \multicolumn{1}{c|}{$5$~\cite{bajpai2022revisiting}} & $3.94$ \\
    \hline Radii are powers of $b$ & \multicolumn{1}{c|}{$\frac{3b-1}{b-1}$~\cite{bajpai2022revisiting}} & $\frac{3b^2-1}{b^2-1}$ \\
    \hline   Priority Knapsack Supplier with Outliers & \multicolumn{1}{c|}{$14$~\cite{bajpai2022revisiting}} & $17$ \\
    \hline   Colorful $k$-Center& & \\
    \hline   \multirow{2}{*}{2 colors} & \multicolumn{1}{c|}{$(2,k+1)$~\cite{bandyapadhyay2019constant}} & \multirow{2}{*}{} \\
    \cline{2-3} & \multicolumn{1}{c|}{$(3,k)$~\cite{jia2022fair}} & \\
    \hline \multirow{2}{*}{$c$ colors} & \multicolumn{1}{c|}{$(2,k+c-1)$~\cite{bandyapadhyay2019constant}} & \multirow{2}{*}{} \\
    \cline{2-3} & \multicolumn{1}{c|}{$(O(1),k)$~\cite{anegg2022techniques}} & \\
    \hline Priority Colorful $k$-Supplier& & \\
    \hline   \multirow{2}{*}{2 colors \& same color, same radius} & \multirow{2}{*}{} & \multicolumn{1}{|c|}{$(2+\sqrt{5},k+1)$} \\
    \cline{2-3} & & \multicolumn{1}{|c|}{$(7,k)$}\\
    \hline $c$ colors & \multicolumn{1}{c|}{} & $(17,k+2c-1)$\\
    \hline
    \end{tabular}
    \caption{Summary of known and new results.}
    \label{table:result}
    \end{center}   
\end{table}


\paragraph{Organization:} Section~\ref{sec:prelims} contains formal definitions of the problems, the LP relaxation we need and a basic filitering technique.
Section~\ref{sec:pkso} describes our improvement for \pkso. Section~\ref{sec:pknapso} describes the Knapsack version of \pkso from \cite{bajpai2022revisiting}, for the sake of completeness, since we need the details of that approach. Section~\ref{sec:pcks} describes our algorithms for \pkso. We conclude in Section~\ref{sec:concl} with a discussion of some open problems.

\section{Preliminaries}
\label{sec:prelims}
\begin{definition}[Priority $k$-Supplier with Outliers (\pkso)]
The input is a metric space $(X=C\cup F,d)$, a radius function $r: C\to \mathbb{R}_{> 0}$, and parameters $k,m\in \mathbb{N}$. The goal is to find $S\subseteq F$ of size at most $k$ to minimize $\alpha$ such that for at least $m$ vertices $v\in C$, $d(v,S)\leq \alpha \cdot r_v$.
\end{definition}

Since there are only polynomially many choices for an optimal $\alpha$, we restrict our attention to the decision version of \pkso and assume that 
the optimal value $\alpha^*$ is equal to $1$ by appropriately scaling the radii. We further assume that for each $v\in C$, there exists some $f\in F$ such that $d(f,v)\leq r_v$, otherwise we have a certificate that $\alpha^* > 1$.

The following is the natural LP relaxation for the decision version for of \pkso with $\alpha=1$. $x_f$ indicates how much a facility $f$ is opened as a center and $\cov(v)$ indicates how much a vertex $v$ is covered.
\begin{align}
    & \sum_{v\in C} \cov(v) \geq m \tag{\pkso LP}\\
    & \sum_{f\in F} x_f \leq k \notag\\
    & \cov(v) = \min(\sum_{\substack{f \in F:\\ d(f,v) \leq r_v}} x_f, 1) & \forall v\in C \notag\\
    & 0 \leq x_f  \leq 1 & \forall f\in F \notag\\
\end{align}

In the preceding formulation the constraint $\cov(v) = \min(\sum_{\substack{f \in F:\\ d(f,v) \leq r_v}} x_f, 1)$ is not linear but can be replaced with two linear
constraints; we keep this notation for ease of understanding.

Filter is a standard procedure in clustering problems to partition the vertex set into several well-separated clusters. In Algorithm~\ref{alg:filter}, the vertices are ordered according to their $\cov$ values, and a set of representatives $R$ and corresponding clusters $\{D(v):v\in R\}$ are obtained. Similar algorithms are also used in \cite{hochbaum1986unified} and \cite{plesnik1987heuristic}: in \cite{hochbaum1986unified}, the vertex ordering is arbitrary, and in \cite{plesnik1987heuristic}, the vertices are sorted based on their radii.

\begin{algorithm}[ht]
    \caption{Filter}
    \label{alg:filter}
	\begin{algorithmic}[1]
		\Require Metric $(X=C\cup F,d)$, radius function $r$, and LP solution $\cov$
		\State $U \leftarrow C$
		\State $R \leftarrow \emptyset$
		\While{$U \neq \emptyset$}
	        \State $v \leftarrow \arg\max_{u\in U} \cov(u)$
	        \State $R \leftarrow R \cup \{v\}$
	        \State $D(v) \leftarrow \{u \in U: d(u,v) \leq r_u + r_v\}$
	        \State $U \leftarrow U \setminus D(v)$
		\EndWhile
		\Ensure $R$, $\{D(v):v \in R\}$
	\end{algorithmic}
\end{algorithm}

\begin{lemma}
\label{lem:filter}
The output of Algorithm~\ref{alg:filter} has the following properties:\\
(a) $\{D(v):v\in R\}$ is a partition of $C$.\\
(b) $\forall u,v\in R, d(u,v)>r_u+r_v$.\\
(c) $\forall v\in R,u\in D(v),d(u,v)\leq r_u+r_v$ and $\cov(v)\geq \cov(u)$.
\end{lemma}

To warm up, we provide a 3-approximation algorithm for $k$-Supplier with Outliers from~\cite{chakrabarty2020non,bajpai2022revisiting}. Note that in the non-priority setting, the radius of every vertex is the same.

\begin{theorem}
\label{thm:k-outlier}
There is a $3$-approximation algorithm for $k$-Supplier with Outliers.
\end{theorem}

\begin{proof}
We first obtain a solution of \pkso LP and run Algorithm~\ref{alg:filter}. Consider the following auxiliary LP where $y_v$ indicates how much $v$ is opened as a center (note that $v$ is a client here, by opening $v$ we mean open some $f\in F$ where $d(f,v)\leq r_v$):
\begin{align*}
\max \quad & \sum_{v\in R} |D(v)|y_v\\
& \sum_{v\in R} y_v \leq k\\
& 0 \leq y_v \leq 1 & \forall v\in R
\end{align*}

We claim that $y_v=\cov(v)$ is a feasible fractional solution with objective value at least $m$. For every $f\in F$, let $A_f=\{v\in C:d(f,v)\leq r_v\}$ be the set of vertices $f$ can cover. Note that $|A_f\cap R|\leq 1$ since otherwise for $u,v\in A_f\cap R$, we have $d(u,v)\leq d(u,f)+d(f,v)\leq r_u+r_v$, which contradicts to Lemma~\ref{lem:filter}(b). Hence
\begin{align*}
    \sum_{v\in R} \cov(v)& = \sum_{v\in R} \sum_{\substack{f \in F:\\ d(f,v) \leq r_v}} x_f\\
    & = \sum_{v\in R}\sum_{f:v\in A_f} x_f\\
    & = \sum_{f\in F} |A_f\cap R| x_f\\
    & \leq \sum_{f\in F} x_f\leq k.
\end{align*}
Hence $y_v=\cov(v)$ is feasible. On the other hand, we have
\begin{align*}
    \sum_{v\in R} |D(v)|\cov(v) 
    & \geq \sum_{v \in R} \sum_{u\in D(v)} \cov(u) & \text{(By Lemma~\ref{lem:filter}(c))}\\
    & = \sum_{v\in C} \cov(v) & \text{(By Lemma~\ref{lem:filter}(a))}\\
    & \geq m.
\end{align*}

The auxiliary LP defines an integer polyhedron since the constraint system corresponds to that of a uniform matroid of rank $k$ (see \cite{Schrijver-book}).
Since there exists a fractional solution with objective value at least $m$, we can obtain an integral solution with objective value at least $m$ by choosing $k$ sets from $\{D(v):v\in R\}$ with largest cardinality. For each chosen $v$, open any facility $f$ such that $d(f,v)\leq r_v$. For all chosen $v$, every $u\in D(v)$ is covered within $3$ times its radius since
\[
    d(u,f)\leq d(u,v)+d(v,f)\leq (r_u+r_v)+r_v = 3r_u.
\]

Therefore, at least $m$ vertices are covered using at most $k$ facilities within 3 times the radius.
\end{proof}

\section{Priority $k$-Supplier with Outliers}
\label{sec:pkso}

In this section, we introduce our improved algorithm utilizing the framework in~\cite{bajpai2022revisiting}, which gives $(1+3\sqrt{3})$-approximation for \pkso. Part of the section is replicating~\cite{bajpai2022revisiting} for the sake of completeness.

\begin{definition}[Weighted $k$-Path Packing (\wkpp)]
The input is a DAG $G=(V,E)$, a value function $\lambda: V\to \{0,1,\dots,n\}$ for some integer $n$, and a parameter $k$. The goal is to find a set of $k$ paths $P\subseteq \mathcal{P}(G)$ that maximizes:
\[
    \val(P)=\sum_{v\in \bigcup_{p\in P} p} \lambda(v).
\]
\end{definition}

The problem is polynomial time solvable by a reduction to Min-Cost Max-Flow. We build a new graph $G'=(V',E')$ based on $G=(V,E)$. For each vertex $v\in V$, split it into two copies $v_1,v_2$. This is to ensure each vertex is count only once even if it is covered by more than one paths.

We set $V'=\bigcup_{v\in V}\{v_1,v_2\} \cup \{s,t\}$ and

\[
E' =  \bigcup_{v\in V} \left\{
\begin{aligned}
&(v_1,v_2)\text{ with capacity } 1 \text{ and cost}-\lambda(v)\\
&(v_1,v_2)\text{ with capacity } \infty \text{ and cost }0\\
&(s,v_1)\text{ with capacity } \infty \text{ and cost }0\\
&(v_2,t)\text{ with capacity } \infty \text{ and cost }0
\end{aligned}
\right\}
\cup \bigcup_{(u,v)\in E} (u_2,v_1)\text{ with capacity } \infty \text{ and cost }0.
\]
Additionlly, the flow passing through sink $t$ is capped by $k$. It is easy to see that \wkpp is equivalent to the Min-Cost Max-Flow from $s$ to $t$ on $G'$.

The following LP is the Min-Cost Max-Flow LP for \wkpp instance. $y_e$ for $e\in E'$ indicates the amount of flow passing through $e$ and $\flow(v)$ for $v\in V$ indicates the amount of flow passing through the arc $(v_1,v_2)$ with capacity $1$ and cost $-\lambda(v)$.
\begin{align}
\max \quad & \sum_{v\in V} \lambda(v)\flow(v)\tag{\wkpp LP}\\
&\sum_{e\in \delta^-(v_1)} y_e = \sum_{e\in \delta^+(v_2)} y_e & \forall v\in V\notag\\
&\flow(v) = \min(\sum_{e\in \delta^-(v_1)} y_e, 1) & \forall v\in V\notag\\
&\flow(t) = \sum_{e\in \delta^-(t)} y_e \leq k\notag\\
& 0 \leq y_e \leq 1 & \forall e\in E'\notag
\end{align}

Now we are ready to present our algorithm to solve \pkso. We partition the vertex sets into several layers $L_1,\dots,L_t$ (in some ways that will be specified later). In each layer $i$, run Algorithm~\ref{alg:filter} on $L_i$ and get the set of representatives $R_i$ and corresponding clusters $\{D(v):v\in R_i\}$. Let $V=\bigcup_{i=1}^t R_i$, we build the contact DAG in the following way.

\begin{definition}[contact DAG]
\label{def:contact-dag}
Contact DAG $G=(V,E)$ is a DAG on vertex set $V$ where each vertex $v\in V$ has a weight $\lambda(v)=|D(v)|$. For $u\in R_i$ and $v\in R_j$ where $i>j$:
\[
    (u,v)\in E \Longleftrightarrow \exists f\in F: d(f,u) \leq r_u \text{ and } d(f,v) \leq r_v.
\]
\end{definition}

By solving \wkpp on the contact DAG, we obtain a collection of $k$ paths $P$. The above steps are summarized in Algorithm~\ref{alg:layering}. Then, we will choose a center to open for each path, and the resulting union of the paths represents the vertices that we cover.

\begin{algorithm}[ht]
    \caption{Find Paths 1}
    \label{alg:layering}
	\begin{algorithmic}[1]
		\Require Metric $(X=C\cup F,d)$, radius function $r$, and LP solution $\cov$
        \For {$i=1$ to $t$}
        \State $R_i,\{D(v):v\in R_i\}\gets$ \text{Filter}$((L_i\cup F,d),r,\cov)$
        \EndFor
        \State Construct contact DAG $G=(V,E)$ according to Definition~\ref{def:contact-dag}
        \State Get a solution $P$ for \wkpp on $G$, i.e. an integral solution to the \wkpp LP
		\Ensure $P$
	\end{algorithmic}
\end{algorithm}

\begin{lemma}
\label{lem:path-cover}
The output of Algorithm~\ref{alg:layering} $P$ satisfies
\[
    \val(P) = \sum_{v\in \bigcup_{p\in P} p} \lambda(v) \geq m.
\]
\end{lemma}

\begin{proof}
Note $P$ implies an integral solution to the \wkpp LP, with $\val(P)$ equals to the objective value. Similar to the proof of Theorem~\ref{thm:k-outlier}, we show that there is a feasible fractional solution to the \wkpp LP with objective value at least $m$. Since the \wkpp LP is integral, thereby we have $\val(P)=\text{optimal integral solution}\geq m$.

We first obtain a solution of \pkso LP.  For any facility $f\in F$, let $A_f=\{v\in V:d(f,v)\leq r_v\}$ be the set of vertices that $f$ can cover. For all $R_i$, $|A_f\cap R_i|\leq 1$ due to Lemma~\ref{lem:filter}(b). Let $p_f$ be the path in contact DAG connecting vertices in $A_f$ and $s,t$ in topological order. We add $x_f$ to all $y_e$ for $e\in p_f$ and increase $\flow(v)$ by $x_f$ for all $v\in A_f$ (taking minimum with 1).

First we claim that this is a feasible solution to \wkpp LP. Clearly it satisfies flow conservation $\sum_{e\in \delta^-(v_1)} y_e = \sum_{e\in \delta^+(v_2)}y_e$ since each time we add values to a path from $s$ to $t$. Also $\flow(t)\leq k$ since $\flow(t)=\sum_{f\in F}x_f \leq k$.

Next we argue that the objective value is at least $m$. Note that
\[
    \flow(v) = \min(\sum_{\substack{f\in F:\\v\in A_f}}x_f, 1) = \min(\sum_{\substack{f\in F:\\d(f,v)\leq r_v}}x_f, 1) = \cov(v).
\]
and thus we have
\begin{align*}
    \sum_{v\in V} \lambda(v)\flow(v)& =\sum_{v\in V} \lambda(v)\cov(v)\\
    &=\sum_{v\in V}|D(v)|\cov(v)\\
    &\geq \sum_{v\in V} \sum_{u\in D(v)} \cov(u)\\
    &=\sum_{v\in C} \cov(v)\geq m.
\end{align*}

Since the \wkpp LP is integral, the optimal integral solution found by Algorithm~\ref{alg:layering} has an objective value at least $m$.
\end{proof}

\begin{theorem}
\label{thm:same-radii}
There is a $\frac{3b^2-1}{b^2-1}$-approximation algorithm if the radii are power of $b$.
\end{theorem}

\begin{proof}

Suppose radii are from $b^0$ to $b^{t-1}$ where $t$ is even, let $B_i=\{v\in X:r_v = b^{i-1}\}$. We partition the vertex set $C$ such that $L_0=B_{t-1},L_1=B_{t-3},\dots,L_{\frac{t}{2}-1}=B_{1},L_{\frac{t}{2}}=B_{0},L_{\frac{t}{2}+1}=B_{2},\dots,L_{t-1}=B_{t-2}$ and run Algorithm~\ref{alg:layering}. According to Lemma~\ref{lem:path-cover}, we have $\sum_{v\in \bigcup_{p\in P}} |D(v)|\geq m$. Therefore, we only need to show that $\bigcup_{v\in \bigcup_{p\in P}} D(v)$ can be covered within $\frac{3b^2-1}{b^2-1}$ times the radius.

For each path $p\in P$ from the output of Algorithm~\ref{alg:layering}, consider the worst case that it has a vertex $u_i$ from every $L_i$. It will be clear right off why this is the worst case. Let $f_i$ be an arbitrary facility covers both $u_{i-1}$ and $u_{i}$, i.e. $d(f_i,u_{i-1})\leq r_{u_{i-1}}$ and $d(f_i,u_i)\leq r_{u_i}$. It always exists by the construction of contact DAG. We open the facility $f_{\frac{t}{2}}$. Fix some $i\leq \frac{t}{2}-1$. For all $v\in D_{u_i}$, we have
\begin{align*}
\frac{d(v,f_{\frac{t}{2}})}{r_v}&\leq \frac{1}{r_v}\big(d(v,u_i)+d(u_i,f_{i+1})+d(f_{i+1},f_{i+2})+\dots+d(f_{\frac{t}{2}-1},f_{\frac{t}{2}})\big)\\
&\leq \frac{1}{r_v}\big((r_v + r_{u_i}) + r_{u_i} + 2\cdot r_{u_{i+1}} + \dots + 2\cdot r_{u_{\frac{t}{2}-1}}\big)\\
& = 3 + 2\cdot (b^{-2} + b^{-4} + \dots + b^{2i+2-t})\\
& = 3 + 2\cdot \frac{b^{-2}-b^{2i-t}}{1-b^{-2}}\\
&\leq 3 + 2\cdot \frac{b^{-2}}{1-b^{-2}} = \frac{3b^2-1}{b^2-1}.
\end{align*}
For $i\geq \frac{t}{2}$ and odd $t$, the proof is analogous.

Note that if $u_i$ is missing from some layers $L_i$, we can always open the facility in the middle. Formally, for largest $i<\frac{t}{2}$ such that $u_i$ exists and smallest $j\geq \frac{t}{2}$ such that $u_j$ exists, open a facility $f$ covers both $u_i$ and $u_j$, i.e. $d(f,u_i)\leq r_{u_i}$ and $d(f,u_j)\leq r_{u_j}$. For any $i$ and $v\in D_{u_i}$, the worst case distance between $v$ and $f$ can only be smaller when some $u_i$'s are missing since we can take shortcuts by skipping the missing layers. This completes the proof.
\end{proof}

\begin{theorem}
There is a $(1+3\sqrt{3})$-approximation algorithm for \pkso.
\end{theorem}

\begin{proof}

For some $b$ that will be chosen later, suppose $t=\lceil \log_{b}r_{max}\rceil$, let $B_i=\{v\in X:b^{i-1}\leq r_v < b^i\}$. We partition the vertex set $C$ such that $L_0=B_{t-1},L_1=B_{t-3},\dots,L_{\frac{t}{2}-1}=B_{1},L_{\frac{t}{2}}=B_{0},L_{\frac{t}{2}+1}=B_{2},\dots,L_{t-1}=B_{t-2}$ and run Algorithm~\ref{alg:layering}. Similar to the proof of Theorem~\ref{thm:same-radii}, we only need to show that $\bigcup_{v\in \bigcup_{p\in P}} D(v)$ can be covered within $1+3\sqrt{3}$ times the radius.

Following the notation in the proof of Theorem~\ref{thm:same-radii}, for each path $p\in P$, we open $f_{\frac{t}{2}}$ and for all $v\in D_{u_i},i\leq \frac{t}{2}-1$, we have
\begin{align*}
\frac{d(v,f_{\frac{t}{2}})}{r_v}&\leq \frac{1}{r_v}\big(d(v,u_i)+d(u_i,f_{i+1})+d(f_{i+1},f_{i+2})+\dots+d(f_{\frac{t}{2}-1},f_{\frac{t}{2}})\big)\\
&\leq \frac{1}{r_v}\big((r_v + r_{u_i}) + r_{u_i} + 2\cdot r_{u_{i+1}} + \dots + 2\cdot r_{u_{\frac{t}{2}-1}}\big)\\
& \leq 1 + 2\cdot b + 2\cdot (b^{-1} + b^{-3} + \dots + b^{2i+3-t})\\
& = 1 + 2\cdot b + 2\cdot \frac{b^{-1}-b^{2i+1-t}}{1-b^{-2}}\\
&\leq 1 + 2\cdot b + 2\cdot \frac{b^{-1}}{1-b^{-2}} = \frac{2b^3+b^2-1}{b^2-1}.
\end{align*}
which attains its minimum $1+3\sqrt{3}$ when $b=\sqrt{3}$.
\end{proof}

We remark that the main procedures of our algorithm are the same as those in the algorithm of \cite{bajpai2022revisiting} (see their Section 3). The main difference is in the way we order the layers. In \cite{bajpai2022revisiting}, they set $L_0=B_0,L_1=B_1,\dots,L_{t-1}=B_{t-1}$ and open a facility on one side. In our case, we open a facility in the middle and place the layers alternately on two sides. This has enabled us to achieve an improved approximation.

When the number of distinct radii is small, we provide a technique that further improves the approximation ratio: we can consider the relationship of the radii of different layers and decide whether contract them or not. We start with the known result for $2$ radii.

\begin{theorem}
\label{thm:2-radii}
There is a $3$-approximation algorithm when there are two different radii.
\end{theorem}
\begin{proof}
Denote by $r_0,r_1$ the 2 radii. We partition the vertex set to $L_0=\{v\in C:r_v=r_0\},L_1=\{v\in C:r_v=r_1\}$ and run Algorithm~\ref{alg:layering}. We only need to show that $\bigcup_{v\in \bigcup_{p\in P}} D(v)$ can be covered within $3$ times the radius.

Consider a path $\{u_0,u_1\}\in P$ with $v_0\in D(u_0),v_1\in D(u_1)$, which is illustrated in Figure~\ref{fig:2-radii}. If we open $f_1$, it is easy to see that $d(v_0,f_1)\leq 3r_0$ and $d(v_1,f_1)\leq 3r_1$.
\end{proof}

\begin{figure}[htb]
\centering
\begin{tikzpicture}[every node/.style = {draw, circle, inner sep = 1pt}]
    \node (u0) at (0,3) {$u_0$};
    \node (u1) at (0,6) {$u_1$};
    \node (f1) at (0,4.5) {$f_1$};
    \node (v0) at (2,3) {$v_0$};;
    \node (v1) at (2,6) {$v_1$};
    \draw (u0) to node[draw = none, left] {$\leq r_0$} (f1);
    \draw (u1) to node[draw = none, left] {$\leq r_1$} (f1);
    \draw (u1) to node[draw = none, yshift=10] {$\leq 2r_1$} (v1);
    \draw (u0) to node[draw = none, yshift=10] {$\leq 2r_0$} (v0);
\end{tikzpicture}
\caption{Illustration of 2 radii}
\label{fig:2-radii}
\end{figure}
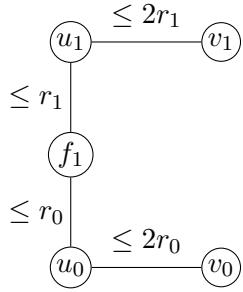

\begin{figure}[hbt]
\centering
\begin{tikzpicture}[every node/.style = {draw, circle, inner sep = 1pt}]
    \node (u0) at (0,3) {$u_0$};
    \node (u1) at (0,0) {$u_1$};
    \node (u2) at (0,6) {$u_2$};
    \node (f1) at (0,1.5) {$f_1$};
    \node (f2) at (0,4.5) {$f_2$};
    \node (v0) at (2,3) {$v_0$};
    \node (v1) at (2,0) {$v_1$};
    \node (v2) at (2,6) {$v_2$};
    \draw (u0) to node[draw = none, left] {$\leq r_0$} (f1);
    \draw (u0) to node[draw = none, left] {$\leq r_0$} (f2);
    \draw (u2) to node[draw = none, left] {$\leq r_2$} (f2);
    \draw (u1) to node[draw = none, left] {$\leq r_1$} (f1);
    \draw (u2) to node[draw = none, yshift=10] {$\leq 2r_2$} (v2);
    \draw (u0) to node[draw = none, yshift=10] {$\leq 2r_0$} (v0);
    \draw (u1) to node[draw = none, yshift=10] {$\leq 2r_1$} (v1);
\end{tikzpicture}
\begin{tikzpicture}[every node/.style = {draw, circle, inner sep = 1pt}]
    \node (u0) at (0,3) {$u_0$};
    \node (u2) at (0,6) {$u_2$};
    \node (f1) at (0,4.5) {$f_1$};
    \node (v0) at (2,3) {$v_0$};;
    \node (v1) at (2,6) {$v_1$};
    \draw (u0) to node[draw = none, left] {$\leq r_0$} (f1);
    \draw (u2) to node[draw = none, left] {$\leq r_2$} (f1);
    \draw (u2) to node[draw = none, yshift=10] {$\leq r_2+r_1$} (v1);
    \draw (u0) to node[draw = none, yshift=10] {$\leq 2r_0$} (v0);
\end{tikzpicture}
\begin{tikzpicture}[every node/.style = {draw, circle, inner sep = 1pt}]
    \node (u1) at (0,3) {$u_1$};
    \node (u2) at (0,6) {$u_2$};
    \node (f1) at (0,4.5) {$f_1$};
    \node (v0) at (2,3) {$v_0$};;
    \node (v2) at (2,6) {$v_2$};
    \draw (u0) to node[draw = none, left] {$\leq r_1$} (f1);
    \draw (u2) to node[draw = none, left] {$\leq r_2$} (f1);
    \draw (u2) to node[draw = none, yshift=10] {$\leq 2r_2$} (v2);
    \draw (u1) to node[draw = none, yshift=10] {$\leq r_1+r_0$} (v0);
\end{tikzpicture}
\caption{Illustration for the case of $3$ distinct radii}
\label{fig:3-radii}
\end{figure}
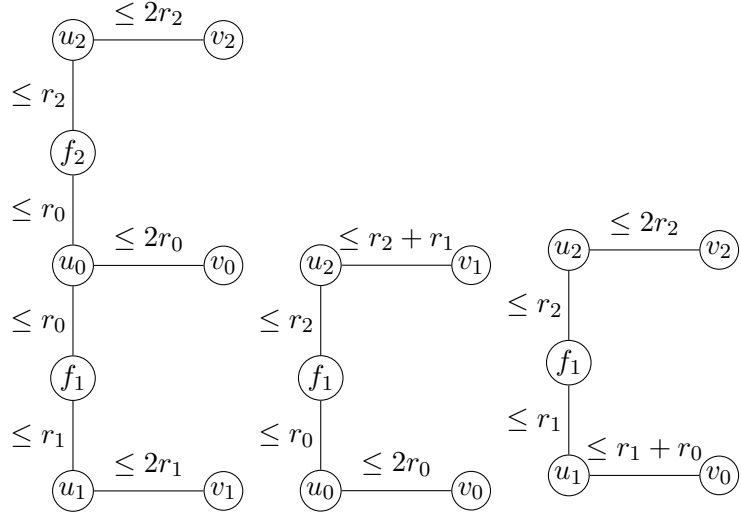

\begin{theorem}
There is a $3.94$-approximation algorithm when there are three different radii.
\end{theorem}
\begin{proof}
Denote by $r_0,r_1,r_2$ the 3 radii and $B_i=\{v\in C:r_v=r_i\}$. As shown in Figure~\ref{fig:3-radii}, consider three ways to partition the vertex set: (a) $L_0=B_0,L_1=B_1,L_2=B_2$; (b) $L_0=B_0,L_1=B_1\cup B_2$; (c) $L_0=B_0\cup B_1,L_1=B_2$. Let $\alpha=\frac{r_1}{r_0},\beta=\frac{r_2}{r_1}$. If we open $f_1$, it is easy to see that the worst case approximation ratio are $\frac{3r_2+2r_0}{r_2}=3+\frac{2}{\alpha\beta},\frac{r_1+2r_2}{r_1}=1+2\alpha,\frac{r_0+2r_1}{r_0}=1+2\beta$ respectively. Given $\alpha,\beta>1$, $\min(3+\frac{2}{\alpha\beta},1+2\alpha,1+2\beta)$ attains its maximum $\approx 3.9311$ at $\alpha = \beta \approx 1.47$.
\end{proof}

We remark that the technique is useful for more than $3$ radii, but we omit this part due to the tedious case analysis.

\section{Priority Knapsack Supplier with Outliers}
\label{sec:pknapso}

In this section, we describe a modification to the algorithm in Section~\ref{sec:pkso} following~\cite{bajpai2022revisiting}, by trading off a slight approximation ratio loss to force the contact DAG to be a forest. In~\cite{bajpai2022revisiting}, the authors claimed a $14$-approximation for \pknapso. However, there is a minor error in Claim 19 of that paper. We reproduce their result in this section to fix the error, and since we also need to use the DAG-to-forest technique in our algorithm for \pcks.

\begin{definition}[Priotrity Knapsack Supplier with Outliers (\pknapso)]
The input is a metric space $(X=C\cup F,d)$, a radius function $r: C\to \mathbb{R}_{>0}$, a weight function $w: F\to \mathbb{R}_{\geq0}$, and parameters $k,m\in \mathbb{N}$. The goal is to find $S\subseteq F$ with $w(S):=\sum_{f\in S}w(f)\leq k$ to minimize $\alpha$ such that for at least $m$ vertices $v\in C, d(v,S)\leq \alpha \cdot r_v$.
\end{definition}

Since the natural LP relaxation of \wknappp has an unbounded integrality gap~\cite{chen2016matroid}, we are going to use an exponential size configuration LP and solve it implicitly using the ellipsoid method. Let $\mathscr{F}:= \{S\subseteq F: w(S)\leq k\}$. Consider the following convex hull of the integral solutions for \pknapso. $z_S$ indicates how much a set of facilities $S\in \mathscr{F}$ is opened as the centers.

\begin{align*}
\pcov = \{ \quad & (\cov(v): v\in C) :\\
& \cov(v)=\min(\sum_{\substack{S\in \mathscr{F}:\\d(v,S)\leq r_v}}z_S,1) & & \forall v\in C,\\
& \sum_{S\in \mathscr{F}} z_S \leq 1,\\
& 0\leq z_S \leq 1 & &\forall S\in \mathscr{F}\quad \}.
\end{align*}

Note that the dimension of $\pcov$ is polynomial although there are exponentially many auxiliary variables $z_S$'s.

\begin{definition}[Weighted Knapsack Path Packing (\wknappp)]
The input is a directed \emph{out-forest} $G=(V,E)$, a value function $\lambda: V\to \{1,2,\dots,n\}$ for some integer $n$, a weight function $w': V\to \mathbb{R}_{\geq0}$ and a parameter $k$. The goal is to find a set of paths $P\subseteq \mathcal{P}(G)$ such that $\sum_{p\in P} w'(\text{sink}(p))\leq k$ that maximizes:
\[
    \val(P)=\sum_{v\in \bigcup_{p\in P} p} \lambda(v).
\]
For $v\in V$, its weight $w'(v)$ is defined as
\[
    w'(v)=\min_{\substack{f\in F:\\d(f,\text{sink}(p))\leq r_{\text{sink}(p)}}} w(f).
\]
\end{definition}

Note that we need a forest instead of a DAG here since solving the \wknappp problem is hard on DAGs. 
When it is a forest one can use a dynamic programming algorithm to solve the problem in polynomial time (c.f. Appendix A of~\cite{bajpai2022revisiting}).

\subsection{DAG to Forest}

To ensure that the contact DAG is a directed out-forest, we use Algorithm~\ref{alg:modified-filter} as a subroutine instead of Algorithm~\ref{alg:filter}. 

\label{sec:dag-to-forest}

\begin{algorithm}[ht]
    \caption{Modified Filter}
    \label{alg:modified-filter}
	\begin{algorithmic}[1]
		\Require Metric $(X=C\cup F,d)$, radius function $r$, LP solution $\cov$, and a parameter $\ell$
		\State $U \leftarrow C$
		\State $R \leftarrow \emptyset$
		\While{$U \neq \emptyset$}
	        \State $v \leftarrow \arg\max_{u\in U} \cov(u)$
	        \State $R \leftarrow R \cup \{v\}$
	        \State $D(v) \leftarrow \{u \in U: d(u,v) \leq r_u + r_v + \ell\}$
	        \State $U \leftarrow U \setminus D(v)$
		\EndWhile
		\Ensure $R$, $\{D(v):v \in R\}$
	\end{algorithmic}
\end{algorithm}

Suppose $t=\lceil \log_{4}r_{max}\rceil$, we partition the vertex set $C$ into $t$ layers such that $L_i=\{v\in C:4^{i-1}\leq r_v < 4^i\}$. In each layer $i$, run Algorithm~\ref{alg:modified-filter} with additional parameter $4^i$ on $L_i$ and get the set of representatives $R_i$ and corresponding clusters $\{D(v):v\in R_i\}$. Let $V=\bigcup_{i=1}^t R_i$, build the contact forest in the following way:

\begin{definition}[contact forest]
\label{def:contact-forest}
Contact DAG $G=(V,E)$ is a DAG on vertex set $V$ where each vertex $v\in V$ has a value $\lambda(v)=|D(v)|$ and a weight $w(v)=\min_{f\in F:d(f,v)\leq r_v}w(f)$. For $u\in R_i$ and $v\in R_j$ where $i>j$:
\[
    (u,v)\in E \Longleftrightarrow d(u,v)\leq r_u+r_v+4^j.
\]
The contact forest is derived by removing all forward edges from $E$, i.e. remove $(u,v)\in E$ if there exists a path from $u$ to $v$ of length greater than 2 in $G$.
\end{definition}

\begin{lemma}
Contact forest is a directed out forest.
\end{lemma}

\begin{proof}
Suppose to the contrary, there are $u\in R_i,v\in R_j,w\in R_k$ where $i\geq j>k$ and $(u,w)\in E,(v,w)\in E$, we have
\begin{align*}
    d(u,v)& \leq d(u,w)+d(w,v)\\
    & = r_u+r_w+4^k+r_v+r_w+4^k\\
    & \leq r_u+r_v+2r_w+2\cdot 4^{k}\\
    & \leq r_u+r_v+4^{k+1}\\
    & \leq r_u+r_v+4^j.
\end{align*}

If $i=j$, this implies $i$ and $j$ should be in the same cluster in Algorithm~\ref{alg:modified-filter} and hence should not be in $R_i$ simultaneously. If $i>j$, this implies $(u,v)\in E$ and hence $(u,w)$ is a forward edge that should be removed.
\end{proof}

\subsection{Round-or-cut}

To solve the LP, we utilize the round-or-cut framework in~\cite{10.1145/3338513}. We use the ellipsoid algorithm on $\pcov$. In each iteration, after getting a solution $\cov$, we either find a feasible collection of paths which implies a approximate \pknapso solution, or we can give the ellipsoid algorithm a separating hyperplane. See Algorithm~\ref{alg:modified-layering}.

\begin{algorithm}[ht]
    \caption{Find Paths 2}
    \label{alg:modified-layering}
	\begin{algorithmic}[1]
		\Require Metric $(X=C\cup F,d)$, radius function $r$
        \State Start an ellipsoid algorithm $\mathcal{E}$ on $\pcov$
        \While {true}
        \State $\{\cov(v):v\in C\} \gets \mathcal{E}$
        \For {$i=1$ to $t$}
        \State $R_i,\{D(v):v\in R_i\}\gets$ \text{Modified Filter}$((L_i\cup F,d),r,\cov, 4^i)$
        \EndFor
        \State Construct contact forest $G=(V,E)$ according to Definition~\ref{def:contact-forest}
        \State Get a solution $P$ for \wknappp on $G$
        \If {$\val(P)\geq m$}
            \State \Return $P$
        \Else
            \State $\mathcal{E}\gets$ a separating hyperplane $\sum_{v\in C}\lambda(v)\cov(v)<m$ \label{line:separate}.
        \EndIf
        \EndWhile
		\Ensure $P$
	\end{algorithmic}
\end{algorithm}

\begin{lemma}[c.f. Lemma 10 of ~\cite{10.1145/3338513}]
\label{lem:in-polytope}
If $\{\lambda(v)\in \mathbb{R}:v\in C\}$ satisfies
\[
    \sum_{\substack{v\in C:\\d(v,S)\leq r_v}} \lambda(v)< m \quad\forall S\in \mathscr{F},
\]
then any $\cov \in \pcov$ satisfies
\[
    \sum_{v\in C} \lambda(v)\cov(v)< m.
\]
\end{lemma}

\begin{proof}
\begin{align*}
\sum_{v\in C} \lambda(v)\cov(v) &\leq \sum_{v\in C} \lambda(v)\sum_{\substack{S\in \mathscr{F}:\\d(v,S)\leq r_v}}z_S\\
& = \sum_{S\in \mathscr{F}} z_S \sum_{\substack{v\in C:\\d(v,S)\leq r_v}} \lambda(v)\\
& < \sum_{s\in \mathscr{F}} z_S\cdot m \leq m.
\end{align*}
\end{proof}

\begin{lemma}[c.f. Lemma 23 of~\cite{bajpai2022revisiting}]
\label{lem:separate}
Each time at Line~\ref{line:separate} of Algorithm~\ref{alg:modified-layering}, we have
\[
    \sum_{\substack{v\in C:\\d(v,S)\leq r_v}} \lambda(v)\leq m \quad\forall S\in \mathscr{F}
\]
and
\[
    \sum_{v\in C} \lambda(v)\cov(v)\geq m
\]
where
\[
    \lambda(v)=\begin{cases}
        |D(v)|& v\in V\\
        0& \text{otherwise}
    \end{cases}.
\]
\end{lemma}
\begin{proof}
Fix any $S\in \mathscr{F}$. For some $f\in S$, let $A_f=\{v\in V:d(f,v)\leq r_v\}$ be the set of vertices $f$ can cover and $p_f$ be the path in contact forest connecting vertices in $A_f$ in topological order. Clearly $P'=\bigcup_{f\in S}p_f$ is a feasible solution to \wknappp. Since $P$ is the optimal solution and $\val(P)<m$, we have
\[
    m>\val(P)\geq \val(P')\geq \sum_{v\in \bigcup_{f\in S}A_f} \lambda(v) = \sum_{\substack{v\in V:\\d(v,S)\leq r_v}} \lambda(v) = \sum_{\substack{v\in C:\\d(v,S)\leq r_v}} \lambda(v).
\]

On the other hand, we have
\begin{align*}
\sum_{v\in C} \lambda(v)\cov(v) & =\sum_{v\in V} \lambda(v)\cov(v)\\
& =\sum_{v\in V} |D(v)|\cov(v)\\
& \geq \sum_{v\in V} \sum_{u\in D(v)} \cov(u)\\
& = \sum_{v\in C} \cov(v)\geq m.
\end{align*}
\end{proof}

\begin{theorem}
\label{thm:knapsack}
There is a $17$-approximation algorithm for \pknapso.
\end{theorem}

\begin{proof}
Combining Lemma~\ref{lem:in-polytope} and Lemma~\ref{lem:separate}, we see that $\sum_{v\in C} \lambda(v)\cov(v)<m$ is indeed a separating hyperplane. Hence the correctness of Algorithm~\ref{alg:modified-layering}.

For the approximation guarantee, we provide a similar proof to Theorem~\ref{thm:same-radii}. For a path $p\in P$, consider the worst case that it has a vertex $u_i$ from each layer $L_i$. Open the facility
\[
    f^*= \arg\min_{\substack{f\in F:\\d(f,\text{sink}(p))\leq r_{\text{sink}(p)}}} w(f).
\]
For a vertex $v\in D(u_i)$, we have

\begin{align*}
\frac{d(v,f^*)}{r_v} &\leq \frac{1}{r_v}\big(d(v,u_i)+d(u_i,u_{i-1})+\dots+d(u_2,u_1)+d(u_1,f^*)\big)\\
&\leq \frac{1}{r_v}\big((r_v+r_{u_i}+4^i)+(r_{u_i}+r_{u_{i-1}}+4^{i-1})+\dots + (r_2+r_1+4)+r_1\big)\\
&\leq \frac{1}{r_v}\big((r_v+2\cdot4^i)+(4^i+2\cdot 4^{i-1})+\dots + (4^2+2\cdot 4)+4\big)\\
&\leq 1 + 3\cdot (4+1+4^{-1}+\dots+4^{2-i})\\
& = 1+ 4\cdot(4-4^{1-i}) \leq 17.
\end{align*}

\end{proof}

\section{Priority Colorful $k$-Supplier}
\label{sec:pcks}

In this section, we discuss the Priority Colorful $k$-Supplier problem.

\begin{definition}[Priority Colorful $k$-Supplier (\pcks)]
The input is a metric space $(X=C\cup F,d)$, a radius function $r:C\to \mathbb{R}_{>0}$, a partition $\{C_1,C_2,\dots,C_c\}$ of $C$ into $c$ colors, a coverage requirement $0\leq m_i\leq |C_i|$ for each color $1\leq i\leq c$, and a parameter $k$. The goal is to find $S\subseteq F$ of size at most $k$ to minimize $\alpha$ such that there are at least $m_i$ vertices $v$ from each color $i$ satisfying $d(v,S)\leq \alpha \cdot r_v$.
\end{definition}

\begin{definition}[\pcks LP]
\begin{align*}
\sum_{v\in C_i} \cov (v) & \geq m_i & 1\leq i\leq c\\
\sum_{f\in F} x_f & \leq k\\
\cov(v) &= \min(\sum_{\substack{f \in F:\\ d(f,v) \leq r_v}} x_f, 1)& \forall v\in C\\
0 \leq x_v & \leq 1 &\forall v\in C\\
\end{align*}
\end{definition}

Follow the construction in Section~\ref{sec:dag-to-forest}, we can obtain a directed out-forest.

\begin{definition}[Weighted Colorful $k$-Path Packing (\wckpp)]
The input is a directed out-forest $G=(V,E)$, for each $1\leq i\leq c$ a value function $\lambda_i: V\to \{0,1,\dots,n\}$ for some integer $n$, and parameters $k$, $m_i$ for $1\leq i \leq c$. The goal is to find a set of $k$ paths $P\subseteq \mathcal{P}(G)$ such that:
\[
    \sum_{v\in \bigcup_{p\in P} p} \lambda_i(v)\geq m_i \quad \forall 1\leq i\leq c.
\]
\end{definition}

Denote by $L$ the leaves in the forest and $T_v$ the subtree rooted at $v$. Since $\lambda_i(v)\geq 0$, we only consider picking the paths from some leaf $v\in L$ to the root. It is clear that we will always choose paths from a leaf to the root.

The following is the natural LP relaxation of the \wckpp problem. $y_v$ for $v\in L$ indicates how much the path from $v$ to the root is chosen. $z_v$ for $v\in V\setminus L$ indicates how much $v$ is covered by the chosen paths. We separately take out the constraint for the first color as the objective value, so that we can save one constraint and thereby save one additional center.

\begin{align}
\max \quad & \sum_{v\in L} \lambda_1(v)y_v + \sum_{v\in V\setminus L} \lambda_1(v)z_v\tag{\wckpp LP}\\
& \sum_{v\in L} \lambda_i(v)y_v + \sum_{v\in V\setminus L} \lambda_i(v)z_v \geq m_i& 2\leq i\leq c\label{constraint:mi}\\
& z_v \leq \sum_{u\in T_v} y_u\label{constraint:yz} & \forall v\in V\setminus L\\
& \sum_{v\in L} y_v\leq k\label{constraint:k}\\
& 0 \leq y_v \leq 1& \forall v\in L\label{constraint:y}\\
& 0 \leq z_v \leq 1& \forall v\in V\setminus L\label{constraint:z}
\end{align}

Similar to the algorithm in Section~\ref{sec:dag-to-forest}, suppose $t=\lceil \log_{4}r_{max}\rceil$, we partition the vertex set $C$ into $t$ layers such that $L_i=\{v\in C:4^{i-1}\leq r_v < 4^i\}$. In each layer $i$, run Algorithm~\ref{alg:modified-filter} with additional parameter $4^i$ on $L_i$ and get the set of representatives $R_i$ and corresponding clusters $\{D(v):v\in R_i\}$. Then build the contact forest according to Definition~\ref{def:contact-forest}.

\begin{algorithm}[ht]
    \caption{Find Paths 3}
    \label{alg:colorful-layering}
	\begin{algorithmic}[1]
		\Require Metric $(X=C\cup F,d)$, radius function $r$, and LP solution $\cov$
        \For {$i=1$ to $t$}
        \State $R_i,\{D(v):v\in R_i\}\gets$ \text{Modified Filter}$((L_i\cup F,d),r,\cov,4^i)$
        \EndFor
        \State Construct contact forest $G=(V,E)$ according to Definition~\ref{def:contact-forest}
        \State Get an extreme point solution $(y^*,z^*)$ to \wckpp LP
        \State $S\gets \{v\in L: y^*_v>0\}$
		\Ensure $S$
	\end{algorithmic}
\end{algorithm}

\begin{lemma}
\label{lem:colorful-km}
There exists a solution to \wckpp LP with objective value at least $m_1$.
\end{lemma}

\begin{proof}
For $f\in F$, let $A_f=\{v\in V:d(f,v)\leq r_v\}$ be the set of vertices $f$ can cover. By the nature of filtering algorithm, $A_f$ contains at most one vertex from each $C_i$. Suppose $v$ is the vertex in $A_f$ with smallest $r_v$. We increase some $y_u<1$ for $u\in T_v\cap L$ by a total amount of $x_f$ and increase $z$ values for the ancestors of $u$ accordingly. All these increase are capped by 1. In this process, it is guaranteed that for all $v\in L$ :
\[
    y_v\geq \sum_{v\in A_f} x_f = \sum_{\substack{f\in F:\\d(f,r)\leq r_v}} x_f=\cov(v).
\]
and for all $v\in V\setminus L$ :
\[
    z_v\geq \sum_{v\in A_f} x_f = \sum_{\substack{f\in F:\\d(f,r)\leq r_v}} x_f=\cov(v).
\]
Hence for any $1\leq i\leq c$, we have

\begin{align*}
\sum_{v\in L} \lambda_i(v)y_v+\sum_{v\in V\setminus L} \lambda_i(v)z_v&\geq
\sum_{v\in V} \lambda_i(v)\cov(v)\\
&= \sum_{v\in V} |D(v)\cap C_i|\cov(v)\\
&\geq \sum_{v\in V}\sum_{u\in D(v)\cap C_i} \cov(u)\\
&= \sum_{v\in C_i} \cov(v)\geq m_i.
\end{align*}

The proof is completed by noting that
\[
    \sum_{v\in L}y_v\leq \sum_{f\in F}x_f\leq k.
\]
\end{proof}

\begin{lemma}
In any extreme point feasible solution to a linear program, the number of linearly independent tight constraints is equal to the number of variables.
\end{lemma}

\begin{lemma}
\label{lem:extreme-point}
For an extreme point solution $(y^*,z^*)$ to \wckpp LP, there are at most $2c$ strictly fractional $y^*_v$.
\end{lemma}

\begin{proof}
In the \wckpp LP, except for the constraints~\ref{constraint:mi} and~\ref{constraint:k}, there are at least $|V|-c$ independent tight constraints.

Denote by $\mathcal{T}$ the set of trees in the forest. For each tree $t\in \mathcal{T}$, let
\[
    \gamma_t = \text{size}(t) - \text{linearly independent tight constraints in } t
\]
where $\text{size}(t)$ denotes the number of vertices in $t$. Note that $\gamma_t\geq 0$ and $\sum_{t\in \mathcal{T}} \gamma_t\leq |V|-(|V|-c)=c$. Denote by $\tau_t$ the number of strictly fractional leaves in $t$. We call a vertex $v\in V\setminus L$ \emph{tight} if constraints~\ref{constraint:yz} and~\ref{constraint:z} for $v$ is tight simultaneously. We pay special attention to tight vertices since we can begin with $\gamma_t=\tau_t$ and $\gamma_t$ can only be decreased by 1 due to the existance of a tight vertex in $t$.

Recall that $L$ denotes the leaves in the forest and $T_v$ denotes the subtree rooted at $v$.

\begin{claim}
If $u,v$ are both tight and $u\in T_v$, then $v$ does not provide an additional linearly independent tight constraint.
\end{claim}
This is because both $u,v$ are tight implies $x_w=0$ for all $w\in (T_v\setminus T_u)\cap L$. The constraint~\ref{constraint:yz} for $v$ is a linear combination of constraints~\ref{constraint:yz} and~\ref{constraint:z} for $u$ and constraint~\ref{constraint:y} for all $w\in (T_v\setminus T_u)\cap L$. Since all of them are tight, $v$ does not provide an additional linearly independent tight constraint.

\begin{claim}
If $v$ is tight, then there are $0$ or at least $2$ fractional leaves in $T_v$.
\end{claim}

Combining the above two claims, we conclude that if we start with $\gamma_t=\tau_t$, it can only be decreased by 1 at a tight vertex $v$ if $\nexists u\in T_v$ such that $u$ is tight and there are at least $2$ fractional leaves in $T_v$. Hence we have $\gamma_t\geq \frac{\tau_t}{2}$ for every $t\in \mathcal{T}$ and therefore
\[
    \sum_{t\in \mathcal{T}} \tau_t\leq 2 \sum_{t\in \mathcal{T}} \gamma_t\leq 2c.
\]

\end{proof}

\begin{theorem}
There is a $17$-approximation algorithm for \pcks using at most $k+2c-1$ centers. 
\end{theorem}

\begin{proof}
For each $v\in S$, the output of Algorithm~\ref{alg:colorful-layering}, we open any facility $f\in F$ such that $d(f,v)\leq r_v$ and covers all clusters on the path from $v$ to the root. By Lemma~\ref{lem:colorful-km}, it satisfies the coverage requirements $m_i$. Also we are using at most $k+2c-1$ centers by Lemma~\ref{lem:extreme-point}. For the approximation guarantee, see the proof of Theorem~\ref{thm:knapsack}.
\end{proof}

In the following of this section, we focus on a special case of \pcks problem. We assume that the vertices with the same color have a same radius. That is, the partition $\{C_1,C_2,\dots,C_c\}$ is
\[
    C_1=\{v\in C:r_v=r_1\}, C_2=\{v\in C:r_v=r_2\}, \dots, C_c=\{v\in C:r_v=r_c\}.
\]
for some $r_1,r_2,\dots,r_c$. We refer to this special case as Uniform Priority Colorful $k$-Supplier (\upcks) problem.

We provide 2 improved algorithms for \upcks when there are only 2 kinds of colors. The first one is a modification of the algorithm for \pcks and the second one follows the framework in~\cite{anegg2022techniques}.

\begin{theorem}
There is a $(2+\sqrt{5})$-approximation algorithm using at most $k+1$ centers for \upcks where there are only 2 colors.
\end{theorem}

\begin{proof}
Let $r_1$ be the radius of vertices in $C_1$ and $r_2$ be that for $C_2$. Assume $r_1<r_2$. If $r_2\leq \frac{1+\sqrt{5}}{2} r_1$, we can use the algorithm introduced in~\cite{bandyapadhyay2019constant} for non-priority case, which yeilds $1+2\cdot \frac{1+\sqrt{5}}{2}=2+\sqrt{5}$ approximation.

If $r_2> \frac{1+\sqrt{5}}{2} r_1$, we use a similar algorithm to Algorithm~\ref{alg:modified-layering}. We run Algorithm~\ref{alg:filter} in $C_1$ and Algorithm~\ref{alg:modified-filter} in $C_2$ with additional parameter $2r_1$. Build the contact forest such that for all $u\in R_2$ and $v\in R_1$,
\[
    (u,v)\in E \Longleftrightarrow \exists f\in F: d(f,u) \leq r_u \text{ and } d(f,v) \leq r_v.
\]
Note that it is a forest since if there exist $u,v\in R_2$ and $w\in R_1$ such that $(u,w)\in E, (v,w)\in E$, then we have
\[
    d(u,v)\leq d(u,w)+d(w,v) = 2\cdot(r_1+r_2) = r_2+r_2+2r_1.
\]
which implies $u$ and $v$ should be in the same cluster.

For each tree of size at least $3$ in the forest, let $v$ be the root and $L$ be the set of leaves. Since $\lambda_2(v)=0$ for all $v\in L$, we can split the tree into a tree of size $2$ formed by $v$ and $\arg\max_{u\in L} \lambda_1(u)$, and $|L|-1$ singletons. Thus we only need to round up at least $2$ fractional leaves and we can use at most $k+1$ centers to satisfy the coverage requirement. The proof for the approximation ratio is analogous to that of Theorem~\ref{thm:2-radii}.
\end{proof}

The following algorithm is based on~\cite{anegg2022techniques}.

\begin{definition}[single color $(L,r)$-partition]
Let $(X=C\cup F, d)$ be a single color metric space. A partition $\mathcal{P}\subseteq 2^{C}$ is an $(L,r)$-partition if
\begin{itemize}
    \item $\text{diam}(A):=\max_{u,v\in A} d(u,v)\leq L\cdot r \quad \forall A\in P$
    \item For any $Z\subseteq F$, there exists a subfamily $\mathcal{A}\subseteq \mathcal{P}$ and injection $h:\mathcal{A}\to Z$ such that
    \begin{itemize}
        \item $d(A,h(A))\leq r \quad \forall A \in \mathcal{A}$
        \item $|\bigcup_{A\in \mathcal{A}} A| \geq |\{v\in C: d(v,Z)\leq r\}|$
    \end{itemize}
\end{itemize}
\end{definition}

\begin{theorem}
There is a $7$-approximation algorithm using at most $k$ centers for \upcks where there are only 2 colors.
\end{theorem}

\begin{proof}
Let $r_1$ be the radius of vertices in $C_1$ and $r_2$ be that for $C_2$. According to~\cite{anegg2022techniques}, we can find a $(6,r_1)$-partition $\mathcal{P}_1$ for $C_1$ and a $(6,r_2)$-partition $\mathcal{P}_2$ for $C_2$. 

Consider a bipartite graph $G=(V=\mathcal{P}_1\cup\mathcal{P}_2,E)$. For any $A_1\in \mathcal{P}_1$ and $A_2\in \mathcal{P}_2$, $(A_1,A_2)\in E$ iff there exists $f\in F$ such that $d(A_1,f)\leq r_1$ and $d(A_2,f)\leq r_2$. We aim to find a matching $M$ (a vertex can be matched to empty) in $G$ such that $|M|\leq k$, $\sum_{A_1\in \mathcal{P}_1,A_1\in M}|A_1|\geq m_1$ and $\sum_{A_2\in \mathcal{P}_2,A_2\in M}|A_2|\geq m_2$. Recall that we are working on the decision version of \upcks, we make the following claim:

\begin{claim}
Such a matching $M$ exists if an optimal solution exists.
\end{claim}

Suppose the optimal solution opens facilities $Z\subseteq F$. According to the property of $(L,r)$-partition, there exists a subfamily $\mathcal{A}_1\subseteq \mathcal{P}_1$ and injection $h_1:\mathcal{A}_1\to Z$, as well as a subfamily $\mathcal{A}_2\subseteq \mathcal{P}_2$ and injection $h_2:\mathcal{A}_2\to Z$. It is easy to see that the following matching is feasible (match to empty if $h_1^{-1}(f)$ or $h_2^{-1}(f)$ doesn't exist)
\[
    M=\bigcup_{f\in Z} (h_1^{-1}(f),h_2^{-1}(f)).
\]

Finding a feasible matching can be reduced the exact weight perfect matching problem. By creating $|\mathcal{P}_2|-k,|\mathcal{P}_1|-k$ dummy vertices on two sides of the bipartite graph respectively and add edges with $\infty$ weights, we can get rid of the constraint $|M|\leq k$. Since $m_1,m_2$ and all $|A_1|,|A_2|$ are polynomial in $n$, we can encode $|A_1|,|A_2|$ into edge weights. By enumerating $\sum_{A_1\in \mathcal{P}_1,A_1\in M}|A_1|$ and $\sum_{A_2\in \mathcal{P}_2,A_2\in M}|A_2|$, we can then determine an exact weight perfect matching, which can be solved in random polynomial time~\cite{maalouly2022exact}.

For every $e=(A_1,A_2)\in M$, open the facility indicated by $e$, i.e. some $f\in F$ such that $d(A_1,f)\leq r_1$ and $d(A_2,f)\leq r_2$. Consider any $v\in A_1$, we have
\[
    d(v,f)\leq \text{diam}(A)+d(A,f) \leq r_1+6r_1 = 7r_1.
\]
and similarly for any $v\in A_2$. Therefore, opening facilities indicated by every $e\in M$ forms a feasible solution within 7 times the radii.
\end{proof}

\section{Conclusions}
\label{sec:concl}
We improved the approximation for $\pkso$ from $9$ to $1+3\sqrt{3}$ via the natural LP. However, the known lower bound on the integrality gap is $3$.
Closing the gap is an interesting open problem. \cite{bajpai2022revisiting} obtained a $9$-approximation even under a matroid constraint on the chosen facilities. 
Is there a better approximation or lower bound for this more general problem? 

We considered \pcks and obtained a bi-criteria approximation that yields a $17$ approximation in the cost while violating the
number of centers by an additive bound of $2c-1$ ($c$ is the number of colors).  Can the approximation bound of $17$ be improved? Is there an $O(1)$ approximation for \pcks that does
not violate the number of centers when $c$ is a fixed constant? For the simpler colorful $k$-center problem there is a $3$-approximation
that does not violate the number of centers when $c$ is fixed constant \cite{jia2022fair,anegg2022}; however the running time of these algorithms is
exponential in $c$ and this is indeed necessary under the exponential time hypothesis \cite{anegg2022}.

\bibliographystyle{alpha}
\bibliography{main}

\end{document}